\documentclass[journal]{IAENGtran}

\usepackage{amsmath, amsthm, amssymb}
\usepackage{amssymb,epsfig}
\usepackage{pstricks} 

\renewcommand{\baselinestretch}{1}
\newcommand{\vs}{2.5pt}  
\newcommand{\et}{1pt}  
\def\qed{\hskip 3pt \hbox{\vrule width4pt depth2pt height6pt}}
\newtheorem{Lemma}{Lemma}

\newtheorem{Theorem}[Lemma]{Theorem}

\newtheorem{Proposition}[Lemma]{Proposition}
\newtheorem{Corollary}[Lemma]{Corollary}

\newcommand{\diam}{\mathop{\mathrm{diam}}\nolimits}
\newcommand{\dist}{\mathop{\mathrm{dist}}\nolimits}

\newcommand{\Fix}{\mathop{\mathrm{Fix}}\nolimits}
\newcommand{\Cay}{\mathop{\mathrm{Cay}}\nolimits}

\begin{document}
%
\title{An Efficient Algorithm for the Diameter of Cayley Graphs Generated by Transposition Trees}

%
%
%

\author{Ashwin~Ganesan
\thanks{Manuscript received June 11, 2012; revised October 4, 2012.

The author is with the Department of Mathematics, Amrita School of Engineering, Amrita University, Coimbatore 641112, Tamil Nadu, India.

Email: ashwin.ganesan@gmail.com


}
}

\maketitle

\pagestyle{empty}
\thispagestyle{empty}

\begin{abstract}
A problem of practical and theoretical interest is to determine or estimate the diameter of various families of Cayley networks. The previously known estimate for the diameter of Cayley graphs generated by transposition trees is an upper bound given in the oft-cited paper of Akers and Krishnamurthy (1989).  In this work, we first assess the performance of their upper bound.  We show that for every $n$, there exists a tree on $n$ vertices, such that the difference between the upper bound and the true diameter value is at least $n-4$.

Evaluating their upper bound takes time $\Omega(n!)$.  In this paper, we provide an algorithm that obtains an estimate of the diameter, but which requires only time $O(n^2)$; furthermore, the value obtained by our algorithm is less than or equal to the previously known diameter upper bound.  Such an improvement to polynomial time, while still performing at least as well as the previous bound, is possible because our algorithm works directly with the transposition tree on $n$ vertices and does not require examining any of the permutations.  We also provide a tree for which the value computed by our algorithm is not necessarily unique, which is an important result because such examples are quite rare. For all families of trees we have investigated so far, each of the possible values computed by our algorithm happens to also be an upper bound on the diameter.

\end{abstract}

\begin{IAENGkeywords}
Cayley graphs; transposition trees; diameter; algorithms;
interconnection networks; permutations.
\end{IAENGkeywords}

%
\IAENGpeerreviewmaketitle

\section{Introduction}

A problem of practical and theoretical interest is to determine or estimate the diameter of various families of Cayley networks. In the field of interconnection networks, Cayley graphs generated by transposition trees were studied in Akers and Krishnamurthy \cite{Akers:Krishnamurthy:1989}, where it was shown that the diameter of some families of Cayley graphs is sublogarithmic in the number of vertices.  This is one of the main reasons such Cayley graphs were considered a superior alternative to hypercubes for consideration as the topology of interconnection networks.  Since then, much work has been done in this area; for further details, we refer the reader to   \cite{Lakshmivarahan:etal:1993}, \cite{Heydemann:1997}, \cite{Du:Hsu:1995}, \cite{Xu:2001}.

Let $G$ be a group generated by a set of elements $S$.  The {\em Cayley graph} (or {\em Cayley diagram}) of $G$ with respect to the set of generators $S$, denoted by $\Cay(G,S)$, is a directed graph with vertex set $G$ and with an arc from vertex $g$ to vertex $gs$ iff $g \in G$ and $s \in S$ (cf. \cite{Bollobas:1998}, \cite{Biggs:1993}).  When $S$ is closed under inverses, there is an arc from $g$ to $h$ if and only if there is an arc from $h$ to $g$, and so $\Cay(G,S)$ can be considered to be an undirected graph. When the identity group element $I$ is not in $S$, the Cayley graph also has no self loops.

Cayley graphs possess a certain amount of symmetry. In a symmetric network, the topology of the network looks the same from every node.  It is now known that many families of symmetric networks possess additional desirable properties such as optimal fault-tolerance \cite{Alspach:1992}, \cite{Akers:Krishnamurthy:1987}, algorithmic efficiency \cite{Annexstein:etal:1990}, optimal gossiping protocols \cite{Bermond:etal:1996} \cite{Zhou:2009}, and optimal routing algorithms \cite{Chen:etal:2006}, among others, and so have been widely studied in the fields of interconnection networks.  New topologies continue to be proposed and assessed \cite{Parhami:2005}.  Cayley graphs, permutation groups, and distance-related questions continue to be an active area of research, and have also found applications in computational biology \cite{Bafna:Pevzner:1998} and wireless sensor networks \cite{Wang:Tang:2007}.

The diameter of a network represents the maximum communication delay between two nodes in the network. The design and performance of algorithms or bounds that determine or estimate the diameter of various families of Cayley graphs of permutation groups is of much theoretical and practical interest. The problem of determining the diameter of a Cayley network is the same as that of determining the diameter of the corresponding group for a given set of generators; the latter quantity is defined to be the minimum length of an expression for a group element in terms of the generators, maximized over all group elements. This diameter problem is difficult even for the simple case when the symmetric group is generated by cyclically adjacent transpositions \cite{Jerrum:1985}.  The pancake flipping problem, which corresponds to determining the diameter of a particular permutation group, was studied in \cite{Gates:Papadimitriou:1979}, and while some bounds were given there and an improvement was made recently  \cite{Chitturi:etal:2009}, this problem remains open as well.

Throughout this paper, set of generates $S$ will be a set of transpositions of $\{1,2,\ldots,n\}$.  Given $S$, the {\em transposition graph} $T(S)$ is defined to be the simple, undirected graph whose vertex set is $\{1,2,\ldots,n\}$ and with vertices $i$ and $j$ being adjacent in $T(S)$ whenever $(i,j) \in S$.   Cayley graphs of permutation groups generated by transposition sets include many topologies of interconnection networks that have been well studied; some special cases include the family of star graphs, bubble-sort graphs, modified bubble-sort graphs, and hypercubes, among others \cite{Heydemann:1997}.

Algorithms for the diameter of Cayley graphs have been widely studied.  When the generator set is fixed in advance, there do exist polynomial time algorithms for the problem of expressing a group element as a product of minimum length in terms of the generators and other distance-related problems (cf. \cite{Jerrum:1985}, \cite{Cooperman:Finkelstein:1992}).  Since a Cayley graph is vertex-transitive, these results provide a polynomial time algorithm for the distance between any two vertices $\pi$ and $\tau$ in the Cayley graph, since their distance is the minimum length generator sequence for $\pi^{-1} \tau$.  However, these results only provide polynomial time algorithms for the distance between one given pair of vertices in the Cayley graph and not for the diameter of the entire Cayley graph, which is the maximum value of the distances between all pairs of vertices of the graph.  Our focus in this paper is on the diameter problem.

A given set of transpositions $S$ of $\{1,2,\ldots,n\}$ generates the entire symmetric group $S_n$ iff the transposition graph $T(S)$ is connected \cite{Godsil:Royle:2001}.  A transposition graph which is a tree is called a transposition tree.   Henceforth, $S$ is a set of transpositions such that the transposition graph $T(S)$ is a tree.  We often use the same symbol $T=T(S)$ to represent both the graph of the tree as well as a set of transpositions $S$, and the notation $(i,j)$ is used to represents both an edge of $T$ as well as the corresponding transposition in $S$.  Since each transposition is its own inverse, the Cayley graph $\Gamma:=\Cay(S_n,S)$ is a simple, undirected graph.  Let $\dist_{\Gamma}(u,v)$ denote the distance between vertices $u$ and $v$ in an undirected graph $\Gamma$, and let $\diam(\Gamma)$ denote the diameter of $\Gamma$.  Note that $\dist_{\Gamma}(\pi,\sigma) = \dist_{\Gamma}(I,\pi^{-1}\sigma)$, where $I$ denotes the identity permutation.  Thus, the diameter of $\Gamma$ is the maximum of $\dist_{\Gamma}(I,\pi)$ over $\pi \in S_n$.

In \cite{Akers:Krishnamurthy:1989} (cf. also \cite[p. 188]{Hahn:Sabidussi:1997}), it is shown that the diameter of $\Gamma:=\Cay(S_n,S)$ is bounded as
$$\diam(\Gamma) \le \max_{\pi \in S_n} \left\{ c(\pi)-n+\sum_{i=1}^n \dist_T(i,\pi(i)) \right\},$$
where the maximum is over all permutations in $S_n$, $c(\pi)$ denotes the number of cycles in the disjoint cycle representation of $\pi$, and $\dist_T$ is the distance function on pairs of vertices of the tree.

Observe that evaluating this upper bound requires $\Omega(n!)$ computations since the quantity in braces above needs to be evaluated for each permutation in $S_n$.  When a bound or algorithm is proposed in the literature, it is often of interest to determine how far away the bound can be from the true value in the worst case, and to obtain more efficient algorithms for estimating the parameters.  In this paper, we assess the performance of the previously known upper bound on the diameter, propose a new $O(n^2)$ algorithm to estimate the diameter of Cayley graphs for any given transposition tree, and we investigate the properties and performance of our algorithm.

\section{Preliminaries, and summary of our main results}

We now recall the previous relevant results and some terminology from the literature as well as summarize our contributions to this problem.

Let $S_n$ denote the set of all permutations of $\{1,2,\ldots,n\}$.  We represent a permutation $\pi \in S_n$ as a linear arrangement, as in $[\pi(1),\pi(2),\ldots,\pi(n)]$, or in cycle notation.  $c(\pi)$ denotes the number of cycles in $\pi$, including cycles of length 1.  Thus, if $\pi = [3,5,1,4,2] = (1,3)(2,5) \in S_5$, then $c(\pi)=3$.  For $\pi,\tau \in S_n$, $\pi \tau$ is the permutation obtained by applying $\tau$ first and then $\pi$.  If $\pi \in S_n$ and $\tau = (i,j)$ is a transposition, then $c(\tau \pi) = c(\pi)+1$ if $i$ and $j$ are in the same cycle of $\pi$, and $c(\tau \pi) = c(\pi)-1$ if $i$ and $j$ are in different cycles of $\pi$ (cf. \cite{Biggs:2003}).  $\Fix(\pi)$ denotes the set of fixed points of $\pi$. We assume throughout that the transposition tree has at least 5 vertices since the problem is easily solved for all smaller trees by using brute force.

Throughout this work, $\Gamma$ denotes the Cayley graph $\Cay(S_n,S)$ generated by a transposition tree $T=T(S)$.  The previous bound on the diameter is as follows:

\begin{Theorem} \label{thm:AK:dist:ubound}  \cite{Akers:Krishnamurthy:1989}  Let $\Gamma:=\Cay(S_n,S)$ be the Cayley graph generated by a transposition tree $T(S)$.  Then, for any $\pi \in S_n$,
$$\dist_{\Gamma}(I,\pi) \le c(\pi)-n+\sum_{i=1}^n \dist_T(i,\pi(i)),$$
where $c(\pi)$ is the number of cycles (including fixed points) in the disjoint cycle representation of $\pi$.
\end{Theorem}

Since $\Gamma$ is vertex-transitive and $\dist_{\Gamma}(\pi,\tau) = \dist_{\Gamma}(I, \pi^{-1} \tau)$, by taking the maximum over both sides of the above inequality, we obtain:
\begin{Corollary} \label{cor:diam:ubound} \cite[p.188]{Hahn:Sabidussi:1997}
$$\diam(\Gamma) \le \max_{\pi \in S_n} \left\{ c(\pi)-n+\sum_{i=1}^n \dist_T(i,\pi(i)) \right\} =: f(T).$$
\end{Corollary}

In the sequel we shall often refer to $f(T)$ as the the previously known upper bound on the diameter of the Cayley graph or {\em the diameter upper bound}. Note that evaluating this estimate $f(T)$ requires time $\Omega(n!)$ since the quantity in braces needs to be evaluated for each of the $n!$ vertices of the Cayley graph.  While one can investigate methods to optimize such an algorithm, at present there is no known polynomial time algorithm for computing $f(T)$.  It is not known whether an algorithm for computing $f(T)$ can be optimized to run in polynomial time or time better than $\Omega(n!)$.  Our objective is to take a different approach and propose a new algorithm that works directly with the transposition tree on $n$ vertices rather than the Cayley graph on $n!$ vertices.

We now recall from \cite{Akers:Krishnamurthy:1989} a proof sketch of these results since we refer to this terminology in the sequel.  Suppose we are given a transposition tree $T$ on vertex set $\{1,2,\ldots,n\}$ and a permutation $\pi \in S_n$ for which we wish to determine $\dist_{\Gamma}(I,\pi)$.  At each vertex $i$ of the tree, we place a marker labeled $\pi(i)$.  Thus, the permutation $\pi$ represents the current position of the markers $1,2,\ldots,n$ on the tree.  To {\em apply an edge} $(i,j)$ of the tree to the current position of markers is to say that we switch the markers at the endpoints of the edge $(i,j)$.  Note that the permutation corresponding to the new position of the markers is exactly $\pi (i,j)$ (here, we read products or compositions of permutations from right to left). The problem of determining $\dist_{\Gamma}(I,\pi)$ is thus equivalent to that of determining the minimum number of edges necessary to `home' each marker $i$ to vertex $i$ of the tree,  and the Cayley graph $\Gamma$ represents the state transition diagram of the current position of markers.  This problem of placing markers $\pi(i)$ at vertex $i$ of an $n$-vertex graph and homing each marker to its vertex is sometimes also referred to as `sorting' a permutation using only the transpositions defined by $T$.  Such problems also arise in routing problems in Cayley networks; for example, a node $\sigma$ receiving a message destined to node $\tau$ of $\Gamma$, or equivalently, a node $\pi = \tau^{-1} \sigma$ receiving a message destined to node $I$, needs to figure out which of its neighbors in $\Gamma$ is closest to the destination node, and this amounts to determining which edge of the tree is optimal in terms of the objective of sorting the current permutation of markers using the minimum number of edges.  Further details can be found in \cite{Akers:Krishnamurthy:1989}.

Let $T$ be a tree on vertex set $\{1,\ldots,n\}$, and let $\pi \ne I$ be the current position of markers on the tree.  It can be shown that $T$ always has an edge $ij$ such that the edge satisfies one of the following two conditions: Either (A) the marker at $i$ and the marker at $j$ will both reduce their distance to $\pi(i)$ and $\pi(j)$, respectively, if the edge $(i,j)$ is applied, or (B) the marker at one of $i$ or $j$ is already homed, and the other marker wishes to apply the edge $(i,j)$.
Let $f_T(\pi)$ denote the upper bound quantity in the right hand side of the inequality in Theorem~\ref{thm:AK:dist:ubound} .  Thus, $f(T) = \max_{\pi \in S_n} f_T(\pi)$.  It can be shown that during each step that a transposition corresponding to an edge of type A or type B is applied to $\pi$, we get a new position of markers  $\pi'$ which has a strictly smaller value of $f_T$; i.e., $f_T(\pi') < f_T(\pi)$, and it can be verified that $f_T(I)=0$.  This proves the bounds above.

We point out that this same diameter upper bound inequality of Corollary~\ref{cor:diam:ubound} is also derived in \cite{Vaughan:1991}; however, this paper was published in 1991, whereas Akers and Krishnamurthy \cite{Akers:Krishnamurthy:1989} was published in 1989 and widely picked up on in the interconnection networks community by then.  There are some subsequent papers, such as \cite{Vaughan:Portier:1995} and \cite{Smith:1999}, which cite only \cite{Vaughan:1991} and not \cite{Akers:Krishnamurthy:1989}.

Note that the distance and diameter bounds above need not hold if $T$ has cycles (the proof recalled above breaks down because if $T$ has cycles, there exists a $\pi \ne I$ such that $T$ has no admissible edges for this $\pi$).  Thus, when we study the strictness of the diameter upper bound, we assume that $T$ is a tree and $\Gamma$ is the Cayley graph generated by this tree.

The exact diameter value of Cayley graphs generated by transposition trees is known in only some special cases.  For example, if the transposition tree is a path graph on $n$ vertices, the corresponding Cayley graph is called a bubble-sort graph.  It is well known that the diameter of this Cayley graph is equal to the maximum number of inversions of a permutation, which is ${n \choose 2}$ (cf. \cite{Akers:Krishnamurthy:1989} \cite{Berge:1971}).  When the transposition tree is a star $K_{1,n-1}$, the Cayley graph is called a star graph, and it has diameter equal to $\lfloor 3(n-1)/2 \rfloor$ (cf. \cite{Akers:Krishnamurthy:1989}).  For the general case of arbitrary trees, only bounds such as Corollary~\ref{cor:diam:ubound} are known.

It is possible to obtain a heuristic derivation of the diameter upper bound formula $f(T)$, as follows.  It is straightforward to derive the distance upper bound for the special case when the transposition tree is a star $K_{1,n-1}$, and we get \cite{Akers:Krishnamurthy:1989}
$$\dist_{\Gamma}(I,\pi) \le n+c(\pi)-2|\Fix(\pi)| - r(\pi).$$
Observe that $|\Fix(\pi)| = n-|\overline{\Fix(\pi)|}$, which yields
$$\dist_{\Gamma}(I,\pi) \le c(\pi) - n + 2|\overline{\Fix(\pi)}| - r(\pi).$$
Note that when the tree is a star, $2 |\overline{\Fix(\pi)}|$ is almost (i.e. within 1 of) the sum of distances $\sum_{i=1}^n \dist_T(i,\pi(i))$.  This leads us to the question of whether the inequality
$$\dist_{\Gamma}(I,\pi) \le c(\pi)-n+\sum_{i=1}^n \dist_T(i,\pi(i))$$
also holds for all the remaining trees $T$, and this question has been answered affirmatively by Theorem~\ref{thm:AK:dist:ubound}.

We shall later use the following result on the sharpness of the diameter upper bound inequality:

\bigskip
\begin{Theorem} \cite{Ganesan:JCMCC} \label{Thm:pathfTmax:equals:diam} Let $\Gamma$ denote the Cayley graph generated by a transposition tree $T$. Then the diameter upper bound inequality
$$\diam(\Gamma) \le f(T)$$
holds with equality if $T$ is a path.
\end{Theorem}

\smallskip The main results of this paper are as follows.

It is of interest to know how far away the diameter upper bound $f(T)$ bound can be from the true value in the worst case.  We show that for every $n$, there exists a transposition tree on $n$ vertices such that the difference between the diameter upper bound and the true diameter value of the Cayley graph is at least $n-4$.  This result gives a lower bound on the difference, and we leave it as an open problem to determine an upper bound for this difference.

We provide a new algorithm (Algorithm A below) which more efficiently computes, for any given transposition tree, an estimate of the diameter of the Cayley graph generated by the tree.  Remarkably, the proposed  algorithm requires only time $O(n^2)$ to compute, whereas no polynomial time algorithm is known for computing the previous bound $f(T)$. Furthermore, it is shown that the value computed by Algorithm A is at least as good as (i.e. is less than or equal to) the previous upper bound $f(T)$.  Such a result is possible because the new algorithm works directly with the transposition tree on $n$ vertices and does not require examining the vertices of the Cayley graph; it is only the proofs of our results that require examining the individual permutations.  It is proved that sometimes the value obtained by Algorithm A is strictly better than (i.e. is strictly less than) the previous upper bound in the literature (cf. Theorem~\ref{thm:beta:lessthan:fT:is:possible}).  Some advantages of the new algorithm over the previous upper bound are illustrated; for example, the proofs related to the worst case performance (of $n-4$) of the new algorithm are much simpler than those of the previous upper bound (cf. Proposition~\ref{prop:strictness:AlgA} and the remarks preceding it).  It is also shown that the value computed Algorithm A is not necessarily unique (Theorem~\ref{thm:nonuniqueness:algA}); this is an important result because such counterexamples are quite rare.

For all families of trees we have investigated so far, each of the possible values $\beta$ computed by Algorithm A is an upper bound on the diameter, i.e.
 $$\diam(\Gamma)~ \le~ \beta ~ \le ~ f(T);$$
here, we prove that the second inequality holds for all trees, and the first inequality holds for many families of trees (in fact for all trees investigated so far).

Some further interesting questions and open problems on this algorithm and related bounds are discussed towards the end of this paper.

\section{Strictness of the diameter upper bound}
Recall that the diameter of a Cayley graph $\Gamma$ generated by a transposition tree $T$ is bounded as
$$\diam(\Gamma) \le \max_{\pi \in S_n} \left\{ c(\pi)-n+\sum_{i=1}^n \dist_T(i,\pi(i)) \right\} =: f(T).$$
We now assess the performance of this bound and derive a strictness result.  The results in this section also appear in the conference paper \cite{Ganesan:ICMMSC:2012}.

Define the worst case performance of this upper bound by the quantity
$$\Delta_n := \max_{T \in \mathcal{T}_n} |f(T) - \diam(\Gamma)|,$$
where $\mathcal{T}_n$ denotes the set of all trees on $n$ vertices.

\begin{Theorem} \label{thm:strictness}
For every $n \ge 5$, there exists a tree $T=T(S)$ on $n$ vertices such that the difference between the actual diameter of the Cayley graph $\diam(\Cay(S_n,S))$ and the diameter upper bound $f(T)$ is at least $n-4$; in other words, $\Delta_n \ge n-4.$
\end{Theorem}

\begin{proof}
Throughout this proof, we let $T$ denote the transposition tree defined by the edge set $\{(1,2),(2,3),\ldots,(n-3,n-2),(n-2,n-1),(n-2,n)\}$, which is shown in Figure~\ref{fig:tree:in:proof}.  For conciseness, we let $d(i,j)$ denote the distance in $T$ between vertices $i$ and $j$. Also, for leaf vertices $i,j$ of $T$, we let $T-\{i,j\}$ denote the tree on $n-2$ vertices obtained by removing vertices $i$ and $j$ of $T$.

Our proof is in two parts.  In the first part we establish that $f(T)$ is equal to ${n \choose 2}-2$.  In the second part we show that the diameter of the Cayley graph generated by $T$ is at most ${{n-1} \choose 2}+1$.  Together, this yields the desired result.
\begin{center}
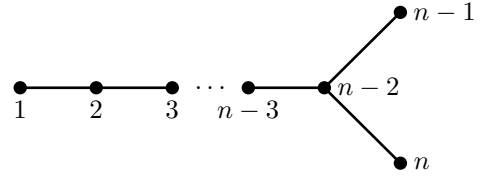
\begin{figure}
\begin{pspicture}(-3,0)(5,2)
\qdisk(0,1){\vs}
\qdisk(1,1){\vs}
\qdisk(2,1){\vs}
\qdisk(3,1){\vs}
\qdisk(4,1){\vs}
\qdisk(5,2){\vs}
\qdisk(5,0){\vs}
\psline[linewidth=\et]{-}(0,1)(1,1)
\psline[linewidth=\et]{-}(1,1)(2,1)
\uput[270](2.5,1.2){$\ldots$}
\psline[linewidth=\et]{-}(3,1)(4,1)
\psline[linewidth=\et]{-}(4,1)(5,2)
\psline[linewidth=\et]{-}(4,1)(5,0)

\uput[270](0,1){$1$}
\uput[270](1,1){$2$}
\uput[270](2,1){$3$}
\uput[270](3,1){$n-3$}
\uput[0](4,1){$n-2$}
\uput[0](5,2){$n-1$}
\uput[0](5,0){$n$}
\end{pspicture}
\caption{A transposition tree $T$ on $n$ vertices.} \label{fig:tree:in:proof}
\end{figure}
\end{center}

We now present the first part of the proof; we establish that $f(T)$, defined by
$$f(T) := \max_{\sigma \in S_n} \left\{ c(\sigma)-n+\sum_{i=1}^n \dist_T(i,\sigma(i)) \right\},$$
is equal to ${n \choose 2}-2$.  We prove this result by examining several sub-cases.  Define
$$f_T(\sigma) := c(\sigma)-n+S_T(\sigma),~~~S_T(\sigma):=\sum_{i=1}^n \dist_T(i,\sigma(i)).$$
We consider two cases, (1) and (2), depending on whether 1 and $n$ are in the same or different cycle of $\sigma$; each of these cases will further involve subcases.  In most of these subcases, we show that for a given $\sigma$, there is a $\sigma'$ such that $f_T(\sigma) \le f_T(\sigma')$ and $f_T(\sigma') \le {n \choose 2}-2$.

(1)  Assume 1 and $n$ are in the same cycle of $\sigma$.  So $\sigma=(1,k_1,\ldots,k_s,n,j_1,\ldots,j_\ell) \hat{\sigma}$.  The different subcases consider the different possible values for $s$ and $\ell$.

(1.1) Suppose $s=0,\ell=0$. So $\sigma=(1,n) \hat{\sigma} = (1,n) \sigma_2 \ldots \sigma_r$.
Then, $f_T(\sigma)=c(\sigma)-n+S_T(\sigma) = r-n+2(n-2)+S_{T-\{1,n\}}(\hat{\sigma})=2n-5+(r-2)+(n-2)+S_{T-\{1,n\}}(\hat{\sigma}) = 2n-5+c(\hat{\sigma})+(n-2)+S_{T-\{1,n\}}(\hat{\sigma}) = 2n-5+f_{T-\{1,n\}}(\hat{\sigma}) \le  2n-5+{{n-2} \choose 2} = {n \choose 2}-2$, where by Theorem~\ref{Thm:pathfTmax:equals:diam} the inequality holds with equality for some $\hat{\sigma}$. Thus, the maximum of $f_T(\sigma)$ over all permutations that contain $(1,n)$ as a cycle is equal to ${n \choose 2}-2$.  It remains to show that for all other kinds of permutations $\sigma$ in the symmetric group $S_n$, $f_T(\sigma) \le {n \choose 2}-2$.

(1.2) Suppose $s=1, \ell=0$.  So $\sigma=(1,i,n) \sigma_2 \ldots \sigma_r = (1,i,n) \hat{\sigma}$. We consider some subcases.

(1.2.1)  Suppose $i=n-1$.  Then, $f_T(\sigma) = r-n+(2n-2)+S_{T-\{1,n-1,n\}}(\hat{\sigma}) = 2n-4+f_{T-\{1,n-1,n\}}(\hat{\sigma}) \le 2n-4+{{n-3} \choose 2} \le {n \choose 2}-2$, where the inequality is by Theorem~\ref{Thm:pathfTmax:equals:diam}.

(1.2.2)  Suppose $2 \le i \le n-2$;  so $\sigma=(1,i,n) \hat{\sigma}$.  Let $\sigma' = (1,n)(i) \hat{\sigma}$.  It is easily verified that $f_T(\sigma) \le f_T(\sigma')$, and so the desired bound follows from applying subcase (1.1) to $f_T(\sigma')$.

 (1.3) Suppose $s=0, \ell=1$, so $\sigma=(1,n,i) \hat{\sigma}$. Since $f_T(\sigma) =  \ f_T(\sigma^{-1})$, this case also is settled by (1.2).

 (1.4)  Suppose $s=0, \ell \ge 2$, so $\sigma = (1,n,j_1,\ldots,j_\ell) \hat{\sigma}$.  Let $\sigma' = (1,n)(j_1,\ldots,j_\ell) \hat{\sigma}$. Observe that $f_T(\sigma) \le f_T(\sigma')$ iff $d(n,j_1)+d(j_\ell,1) \le d(n,1)+d(j_\ell,j_1)+1$.  We prove the latter inequality by considering 4 subcases:

(1.4.1) Suppose $j_1 < j_\ell \le n-2$.  Then, an inspection of the tree in Figure~\ref{fig:tree2:in:proof} shows that $d(n,j_1)+d(j_\ell,1) = d(n,1)+d(j_\ell,j_1)$, and so the inequality holds.
\begin{center}
\begin{figure}
\begin{pspicture}(-3,0)(5,2)
\qdisk(0,1){\vs}
\qdisk(1,1){\vs}
\qdisk(2,1){\vs}
\qdisk(3,1){\vs}
\qdisk(4,1){\vs}
\qdisk(5,2){\vs}
\qdisk(5,0){\vs}
\psline[linewidth=\et]{-}(0,1)(1,1)
\uput[270](0.5,1.2){$\ldots$}
\uput[270](1.5,1.2){$\ldots$}
\uput[270](3.5,1.2){$\ldots$}
\psline[linewidth=\et]{-}(4,1)(5,2)
\psline[linewidth=\et]{-}(4,1)(5,0)

\uput[270](0,1){$1$}
\uput[270](1,1){$2$}
\uput[270](2,1){$j_1$}
\uput[270](3,1){$j_\ell$}
\uput[0](4,1){$n-2$}
\uput[0](5,2){$n-1$}
\uput[0](5,0){$n$}
\end{pspicture}
\caption{Positions of $j_1$ and $j_\ell$ arising in subcase (1.4.1).} \label{fig:tree2:in:proof}
\end{figure}
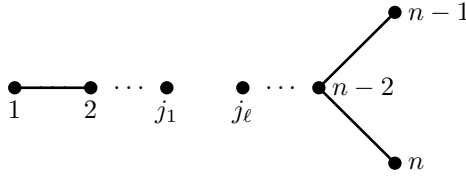
\end{center}
(1.4.2)  Suppose $j_1 > j_\ell$ and $j_1,j_\ell \le n-2$.  Then, $d(n,j_1)+d(j_\ell,1) \le d(1,n)$, and so the inequality holds.

(1.4.3) Suppose $j_1=n-1$.  Then $d(n,j_1)=2$.  Also, $d(j_\ell,1) \le d(n,1)$ and $d(j_\ell,j_1) \ge 1$, and so the inequality holds.

(1.4.4)  Suppose $j_\ell = n-1$.  Then, $d(j_\ell,1)=d(n,1)$ and $d(n,j_1) = d(j_\ell,j_1)$, and so again the inequality holds.

 (1.5) Suppose $s=1, \ell=1$, so $\sigma=(1,i,n,j) \hat{\sigma}$.  Let $\sigma'=(1,n)(i,j) \hat{\sigma}$.

 (1.5.1)  If $i=n-1$, by symmetry in $T$ between vertices $n$ and $n-1$, this subcase is resolved by subcase (1.4).

 (1.5.2) Let $2 \le i \le n-2$.  Then $d(1,i)+d(i,n)=d(1,n)$. So $f_T(\sigma) \le f_T(\sigma')$ iff $d(n,j)+d(j,1) \le d(1,n)+d(i,j)+d(j,i)+1$, which is true since $d(n,j)+d(j,1) \le d(1,n)+2$.

(1.6) Suppose $s=1, \ell \ge 2$. So $\sigma=(1,i,n,j_1,\ldots,j_\ell)\hat{\sigma}$.  Let $\sigma' = (1,n)(i,j_1,\ldots,j_\ell) \hat{\sigma}$.  It suffices to show that $f_T(\sigma) \le f_T(\sigma')$, i.e., that $d(1,i)+d(i,n)+d(n,j_1)+d(j_\ell,1) \le d(1,n)+d(1,n)+d(i,j_1)+d(j_\ell,i)+1$.  We examine the terms of this latter inequality for various subcases:

(1.6.1) Suppose $2 \le i \le n-2$.  Then $d(1,i)+d(i,n)=d(1,n)$.

(1.6.1a)  If $j_\ell = n-1$, then $d(j_\ell,i)=n-i-1$ and $d(i,j_1)=|i-j_1|$, and so the inequality holds iff $-1 \le |j_1-i|+j_1-i$, which is clearly true.

(1.6.1b) Suppose $2 \le j_\ell \le n-2$.  Then, the inequality holds iff $d(n,j_1)+j_\ell-1 \le n-2+|i-j_1|+|i-j_\ell|+1$, which can be verified separately for the cases $j_1=n-1$ and $2 \le j_1 \le n-2$.

(1.6.2) Suppose $i=n-1$.  By symmetry in $T$ of the vertices $n$ and $n-1$, this case is resolved by (1.4).

(1.7) Suppose $s \ge 2, \ell=0$, so $\sigma=(1,k_1,\ldots,k_s,n) \hat{\sigma}$.  Since $f_T(\sigma)=f_T(\sigma^{-1})$, this case is resolved by (1.4).

(1.8) Suppose $s \ge 2, \ell=1$, so $\sigma=(1,k_1,\ldots,k_s,n,j_1) \hat{\sigma}$.  Let $\sigma'=(1,n)(k_1,\ldots,k_s,j_1) \hat{\sigma}$.   We can assume that $2 \le k_1 \le n-2$ since the $k_1=n-1$ case is resolved by (1.4) due to the symmetry in $T$. To show $f_T(\sigma) \le f_T(\sigma')$, it suffices to prove the inequality $d(1,k_1)+d(k_s,n)+d(n,j_1)+d(j_1,1) \le d(1,n)+d(n,1)+d(k_s,j_1)+d(j_1,k_1)+1$.  We prove this inequality by separately considering whether $j_1=n-1$ or $k_s=n-1$ or neither:

(1.8.1) Suppose $j_1=n-1$.  Substituting $d(k_s,n)=n-k_s-1, d(n,j_1)=2, d(j_1,1)=j_1-1$, etc, we get that the inequality holds iff $k_1 \le |j_1-k_1|+n-2$, which is clearly true.

(1.8.2)  Suppose $2 \le j_1 \le n-2$.  Then $d(n,j_1)=n-j_1-1$, and so the inequality holds iff $k_1+d(k_s,n)+n-3 \le 2n-3+d(k_s,j_1)+d(j_1,k_1)$.  If $k_s=n-1$, this reduces to $j_1+k_1 \le 2n-3+|j_1-k_1|$, and is true, whereas if $2 \le k_s \le n-2$, this reduces to $k_1-k_s \le 1+|j_1-k_s|+|j_1-k_1|$, which is true due to the triangle inequality.

(1.9) Suppose $s,l \ge 2$, so $\sigma=(1,k_1,\ldots,k_s,n,j_1,\ldots,j_\ell) \hat{\sigma}$.

Let $\sigma'=(1,k_1,\ldots,k_s,n)(j_1,\ldots,j_\ell) \hat{\sigma}$.  It suffices to show that $S_T(\sigma) \le S_T(\sigma')+1$, i.e., that $d(n,j_1)+j_\ell \le n+d(j_1,j_\ell)$.

(1.9.1) If $j_1 < j_\ell$, then $j_1 \le n-2$, and so $d(n,j_1) =n-j_1-1$ and $d(j_1,j_\ell)=j_\ell-j_1$; the inequality thus holds.

(1.9.2) If $j_1 > j_\ell$, then $d(j_1,j_\ell)=j_1-j_\ell$, and so it suffices to show that $d(n,j_1) \le n+j_1-2j_\ell$.  It can be verified that this holds if $j_1=n-1$ and also if $2 \le j \le n-2$.

(2) Now suppose 1 and $n$ are in different cycles of $\sigma$.  So let $\sigma=(1,k_1,\ldots,k_s)(n,j_1,\ldots,j_\ell) \hat{\sigma}$.

(2.1) Suppose $s=0$.  Then $f_T(\sigma) \le {{n-1} \choose 2}-2$, by induction on $n$.

(2.2) Suppose $s=1$.  So let $\sigma=(1,i)(n,j_1,\ldots,j_\ell)\hat{\sigma}$.  By symmetry in $T$ between vertices $n$ and $n-1$ and subcase (1.1), we may assume $i \ne n-1$. Let $\sigma' = (1,n)(i,j_1,\ldots,j_\ell)\hat{\sigma}$.  It suffices to show that $S_T(\sigma) \le S_T(\sigma')$.  If $\ell=0$ this is clear since $d(1,i) \le d(1,n)$.  Suppose $\ell \ge 2$.  Then, by the triangle inequality, $d(n,j_1)+d(j_\ell,n) \le d(j_1,i)+d(i,n)+d(i,j_\ell)+d(i,n) = d(j_1,i)+d(i,j_\ell)+(n-i-1)2$.  Also, $d(1,n)=d(1,i)+d(i,n) = d(1,i)+n-i-1$.  Hence, $2d(1,i)+d(n,j_1)+d(j_\ell,n) \le 2d(1,n)+d(i,j_1)+d(j_\ell,i)$.  Hence, $S_T(\sigma) \le S_T(\sigma')$.  The case $\ell=1$ can be similarly resolved by substituting $j_1$ for $j_\ell$ in the $l\ge2$ case here.

(2.3) Suppose $s \ge 2, \ell=0$.  Then, by Theorem~\ref{Thm:pathfTmax:equals:diam}, $f_T(\sigma) \le {{n-1} \choose 2}$.

(2.4) Suppose $s \ge 2, \ell=1$, so $\sigma=(1,k_1,\ldots,k_s)(n,j_1) \hat{\sigma}$.  Let $\sigma'=(1,n)(k_1,\ldots,k_s,j_1)\hat{\sigma}$.  It suffices to show that $d(1,k_1)+d(1,k_s)+2d(n,j_1) \le 2d(1,n)+d(k_s,j_1)+d(k_1,j_1)$.  This inequality is established by considering the two subcases:

(2.4.1)  Suppose $j_1=n-1$.  Then the inequality holds iff $2k_s+k_1 \le 3n-7+|k_1-j_1|$, which is true since $k_1,k_2 \le n-2$ and $|k_1-j_1| \ge 1$.

(2.4.2) Suppose $j_1 \ne n-1$.  Then the inequality holds iff $k_1-j_1+k_s-j_1 \le |k_1-j_1|+|k_s-j_1|$, which is clearly true.

(2.5) Suppose $s, \ell \ge 2$, so $\sigma=(1,k_1,\ldots,k_s)(n,j_1,\ldots,j_\ell)\hat{\sigma}$.

Let $\sigma'=(1,n)(k_1,\ldots,k_s,j_1,\ldots,j_\ell)\hat{\sigma}$.  To show $f_T(\sigma) \le f_T(\sigma')$, it suffices to show that $d(1,k_1)+d(k_s,1)+d(n,j_1)+d(j_\ell,n) \le 2d(n,1)+d(k_s,j_1)+d(j_\ell,k_1)$.  By symmetry in $T$ between vertices $n$ and $n-1$, we may assume $k_1,\ldots,k_s \ne n-1$, since these cases were covered in (1).  We establish this inequality as follows:

(2.5.1)  Suppose $j_1=n-1$.  Then $d(n,j_1)=2$ and $d(n,j_\ell)=n-j_\ell-1$.  So the inequality holds iff $2k_s \le 2(n-2)+|j_\ell-k_1|+j_\ell-k_1$, which is true since $k_s \le n-2$ and $|j_\ell-k_1|+j_\ell-k_1 \ge 0$.

(2.5.2)  Suppose $j_1 \ne n-1$.  Then $d(n,j_1)=n-j_1-1$.  If $j_\ell=n-1$, the inequality holds iff $2 k_1 \le 2(n-2)+j_1-k_s+|j_1-k_s|$,which is true since $k_1 \le n-2$.  If $j_\ell \ne n-1$, the inequality holds iff $k_s-j_1+k_1-j_\ell \le |k_s-j_1|+|k_1-j_\ell|$, which is true.

This concludes the first part of the proof.

We now provide the second part of the proof.  Let $\Gamma$ be the Cayley graph generated by $T$.  We show that $\diam(\Gamma) \le {{n-1} \choose 2}+1.$  Let $\pi \in S_n$, and suppose each vertex $i$ of $T$ has marker $\pi(i)$.  We show that all markers can be homed using at most the proposed number of transpositions.  Since $\diam(T)=n-2$, marker 1 can be moved to vertex 1 using at most $n-2$ transpositions.  Now remove vertex 1 from the tree $T$, and repeat this procedure for marker 2, and then for marker 3, and so on, removing each vertex from $T$ after its marker is homed.  Continuing in this manner, we eventually arrive at a star $K_{1,3}$, whose Cayley graph has diameter 4.  Hence, the diameter of $\Gamma$ is at most $[(n-2)+(n-3)+\ldots+5+4+3]+4 = {{n-1} \choose 2}+1$.  This completes the proof.
\end{proof}

Let $s(n)$ denote the number of non-isomorphic trees on $n$ vertices.  Let $\Delta_n$ be the strictness as defined above.   Then, computer simulations yield the results in Table~\ref{table:summary}:
\begin{table}[ht]
\caption{Strictness of the diameter upper bound} 
\centering 
\begin{tabular}{c | c c c c c} 
$n$ & 5 & 6 & 7 & 8 & 9 \\
\hline 
$s(n)$ & 3 & 6 & 11 & 23 & 47 \\
$\Delta_n $ & 1 & 2 & 3 & 4 & 6 \\
\end{tabular}
\label{table:summary} 
\end{table}

These results imply that the $n-4$ lower bound for $\Delta_n$ is best possible, and an open problem is to obtain an upper bound for $\Delta_n$.  That this lower bound is exact for $n=5$ can also be obtained using the results given above, as follows.  There are only three nonisomorphic trees on 5 vertices, namely the tree given in the proof of Theorem~\ref{thm:strictness}, for which the $n-4$ lower bound is achieved, and the trees of maximum diameter (the path) and minimum diameter (the star), for which the diameter bound $f(T)$ is known to be exact (cf. \cite{Ganesan:JCMCC}).

\section{The algorithm}
\label{sec:algorithm}

We now provide an algorithm that takes as its input a transposition tree $T(S)$ on $n$ vertices and provides as output an estimate of the diameter of the Cayley graph $\Cay(S_n,S)$.   The notation used to describe our algorithm should be self-explanatory and is similar to that used in Knuth \cite{Knuth:2011}.

\bigskip
\noindent \textbf{Algorithm A}
\\Given a transposition tree $T=T(S)$, this algorithm computes a value $\beta$ which is an estimate for the diameter of the Cayley graph $\Cay(S_n,S)$.  $|V(T)|$ denotes the current value of the number of vertices in $T$; initially, $V(T)=\{1,2,\ldots,n\}$.
\\ \textbf{A1.} [Initialize.]  \\Set $\beta \leftarrow 0$.
\\ \textbf{A2.} [Find two vertices $i,j$ of $T$ that are a maximum distance apart.]  \\Find any two vertices $i,j$ of $T$ such that $\dist_T(i,j)=\diam(T)$.
\\ \textbf{A3.} [Update $\beta$, and remove $i,j$ from $T$.] \\Set $\beta \leftarrow \beta+(2 \diam(T)-1)$, and set $T \leftarrow T-\{i,j\}$.  If $T$ still has 3 or more vertices, return to step A2; otherwise, set $\beta \leftarrow \beta+|V(T)|-1$ and terminate this algorithm. \qed

\bigskip \textbf{Example 1.} Consider the transposition tree $\{(1,2),(2,3),(3,4),(4,5),(4,6),(3,7),(7,8)\}$ shown in Figure~\ref{fig:tree:T1}.  If Algorithm A picks the sequence of vertex pairs during step A2 to be $\{1,8\},\{5,7\}$ and $\{2,6\}$, then the value returned by the algorithm is $\beta = 7+5+5+1=18$.  On the other hand, if Algorithm A picks the vertex pairs to be $\{1,5\},\{6,8\}$ and $\{2,7\}$, then the value returned by the algorithm is still $\beta=7+7+3+1=18$.  In this example, the value returned by the algorithm is unique even though the subtrees $T-\{1,8\}$ and $T-\{1,5\}$ are non-isomorphic and even have different diameters. \qed
\begin{center}
\begin{figure}
\begin{pspicture}(-3,-0.5)(5,2.5)
\qdisk(0,1){\vs}
\qdisk(1,1){\vs}
\qdisk(2,1){\vs}
\qdisk(3,0){\vs}
\qdisk(3,2){\vs}
\qdisk(4,2.5){\vs}
\qdisk(4,1.5){\vs}
\qdisk(4,0){\vs}
\psline[linewidth=\et]{-}(0,1)(1,1)
\psline[linewidth=\et]{-}(1,1)(2,1)
\psline[linewidth=\et]{-}(2,1)(3,0)
\psline[linewidth=\et]{-}(2,1)(3,2)
\psline[linewidth=\et]{-}(3,2)(4,2.5)
\psline[linewidth=\et]{-}(3,2)(4,1.5)
\psline[linewidth=\et]{-}(3,0)(4,0)
\uput[270](0,1){$1$}
\uput[270](1,1){$2$}
\uput[270](2,1){$3$}
\uput[270](3,0){$7$}
\uput[90](3,2){$4$}
\uput[0](4,2.5){$5$}
\uput[0](4,1.5){$6$}
\uput[0](4,0){$8$}
\end{pspicture}
\caption{The transposition tree on 8 vertices in Example 1.} \label{fig:tree:T1}
\end{figure}
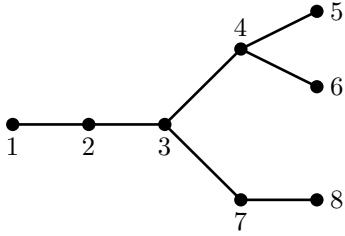
\end{center}

Despite the outcome in the above example where the value computed by Algorithm A is independent of the vertex pairs chosen in step A2, it is shown later that there do exist trees for which the value computed by the algorithm is not unique (i.e. the final value depends on which vertex pairs were chosen during step A2), though this non-uniqueness property is rare.

The eccentricity of a vertex $u$ in a graph is defined to be the maximum value of the distance from $u$ to a vertex of the graph. The center of a graph is defined to be the set of vertices of minimum eccentricity.  It is known that the center of a tree is either a single vertex or two adjacent vertices.  Also, every path of maximum length in a tree passes through its center.

In view of this, one way to implement step A2, which picks any two vertices of the tree that are a maximum distance apart, is as follows.  Start with an arbitrary vertex $u$ of the tree, and do a depth-first search to find a vertex $i$ farthest that is from $u$ ($i$ and $u$ will be on different `sides' of the center).  Then start at vertex $i$ and do another depth-first search to find a vertex $j$ that is farthest from $i$.  Then, the resulting $i-j$ path has maximum length in the tree.  (Alternatively, one could also carry out a breadth-first search rather than a depth-first search.)  Since a depth-first search on the transposition tree on $n$ vertices takes time $O(n)$, the first iteration of Algorithm A takes time $O(n)$, the second iteration takes time $O(n-2)$, and so on.  Thus, Algorithm A takes time $O(n^2)$.  The previously known estimate in the literature for the diameter of Cayley graphs generated by transposition trees works on each of the $n!$ vertices of the Cayley graph and takes time $\Omega(n!)$ to evaluate; hence, Algorithm A, is a significant improvement over the previous bound.

\section{Properties and performance of the algorithm}

In this section we prove that the value obtained by Algorithm A is less than or equal to the previously known diameter upper bound.  We also show that there exist instances where the diameter estimate computed by Algorithm A is strictly better than (i.e. is strictly less than) than the previously known diameter upper bound $f(T)$.  Furthermore, we construct a rare transposition tree for which the value computed by Algorithm A is not unique.

\begin{Theorem} \label{thm:beta:le:AKbound}
Let $T$ be a transposition tree on vertex set $\{1,2,\ldots,n\}$, and let $\beta$ be the value obtained by Algorithm A for this tree. Then, $\beta$ is less than or equal to the previously known upper bound on the diameter of the Cayley graph, i.e.
$$\beta \le f(T).$$
\end{Theorem}

\begin{proof}
Let $\{i_1,j_1\},\{i_2,j_2\},\ldots,\{i_r,j_r\}$ be the vertex pairs chosen by Algorithm A during the $r$ iterations of step A2, where $r= \lfloor (n-1)/2 \rfloor$.  Construct a permutation $\pi$ as follows.  If $n$ is odd, then $T$ contains only one vertex, $i_{r+1}$ say, when the algorithm terminates.  In this case, we let $\pi = (i_1,j_1) \ldots (i_r,j_r)(i_{r+1}) \in S_n$.  If $n$ is even, then $T$ contains two vertices, $i_{r+1}$ and $j_{r+1}$ say, when the algorithm terminates.  In this case, we let $\pi=(i_1,j_1) \ldots (i_{r+1},j_{r+1}) \in S_n$.  In either case, $r+1 = \lceil n/2 \rceil$, the value of $\beta$ computed by the algorithm equals
$$\beta=\left(\sum_{\ell=1}^r \left\{2 \dist_T(i_\ell,j_\ell)-1 \right\} \right)+\left\{ (n+1) \mod 2 \right\},$$
and the quantity $f_T(\pi) := c(\pi)-n+\sum_{i=1}^n \dist_T(i,\pi(i))$ evaluates to
$$f_T(\pi)=(r+1)-n+ \left(2 \sum_{\ell=1}^r \dist_T(i_\ell,j_\ell) \right)+2 \left\{(n+1) \mod 2 \right\}.$$
A quick check shows that the expressions above for $\beta$ and $f_T(\pi)$ are equal.  Hence, for every sequence of vertex pairs chosen by Algorithm A, there exists a permutation $\pi$ such that the value $\beta$ returned by Algorithm A is at most $f_T(\pi)$.  Hence, $\beta \le \max_{\pi \in S_n} f_T(\pi)$.
\end{proof}

Since each of the possible values $\beta$ computed by Algorithm A is at most the previous upper bound $f(T)$, it follows immediately that the largest of the possible values computed by Algorithm A, denoted by $\beta_{\max}$, is also at most the previous upper bound.  We now show that $\beta_{\max}$ is an upper bound on the diameter the Cayley graph:

\begin{Theorem} \label{thm:betamax:ubound}
Let $\Gamma$ be the Cayley graph generated by a transposition tree $T$.  Let $\beta_{\max}$ denote the maximum possible value returned by Algorithm A for this tree.  Then,
$$\diam(\Gamma) ~\le~ \beta_{\max}~ \le~ f(T).$$
\end{Theorem}

\begin{proof}
The second inequality has already been proved.  We now prove the first inequality.  Let $\pi \in S_n$.  Suppose that initially each vertex $k$ of the tree has marker $\pi(k)$. It suffices to show that all markers can be homed to their respective vertices using at most $\beta_{\max}$ edges of the tree.

Consider the following procedure.  Pick any two vertices $i,j$ of $T$ that are a maximum distance apart. We consider two cases, depending on the distance in $T$ between vertex $i$ and the current location $\pi^{-1}(i)$ of the marker $i$:

\textbf{Case 1:} Suppose that the distance in $T$ between vertices $i$ and $\pi^{-1}(i)$ is at most $\diam(T)-1$.  Then marker $i$ can be homed using at most $\diam(T)-1$ transpositions.  And then, marker $j$ can be homed using at most $\diam(T)$ edges.  Hence, markers $i$ and $j$ can both be homed to leaf vertices $i$ and $j$, respectively, using at most $2 \diam(T)-1$ edges.  We now let $i_1=i$ and $j_1=j$.

\textbf{Case 2:} Now consider the case where the distance in $T$ between vertices $i$ and $\pi^{-1}(i)$ is equal to $\diam(T)$.  Let $x$ be the unique vertex of the tree adjacent to $\pi^{-1}(i)$.  In the first sequence of steps, marker $\pi^{-1}(i)$ can be homed to vertex $\pi^{-1}(i)$ using at most $\diam(T)$ edges.  The last of these transpositions will home marker $\pi^{-1}(i)$ and place marker $i$ at vertex $x$, whose distance to $i$ is exactly $\diam(T)-1$.  In the second sequence of steps, marker $i$ can be homed to vertex $i$ using at most $\diam(T)-1$ edges.  Hence, using these two sequences of steps, markers $i$ and $\pi^{-1}(i)$ can both be homed using at most $2 \diam(T)-1$ edges.  We now let $i_1=i$ and $j_1=\pi^{-1}(i)$.

We now remove from $T$ the vertices $i_1$ and $j_1$, and repeat this procedure on $T-\{i_1,j_1\}$ to get another pair $\{i_2,j_2\}$.  Continuing in this manner until $T$ contains at most two vertices, we see that all markers can be homed using at most
$$ \left\{ \sum_{\ell=1}^r \left( 2 \dist_T(i_\ell,j_\ell)-1 \right)\right\} + \left\{(n+1) \mod 2 \right\}$$
edges.  This quantity is equal to the value $\beta$ returned by the Algorithm when it chooses $\{i_1,j_1\},\ldots,\{i_r,j_r\}$ as its vertex pairs during each iteration of step A2, and hence this quantity is at most $\beta_{\max}$.  Thus, $\dist_{\Gamma}(I,\pi) \le \beta_{\max}$ for all $\pi \in S_n$.
\end{proof}

Let $\mathcal{B}$ denote the set of possible values that can be the output of Algorithm A.  Thus, $\beta_{\max} := \max_{\beta \in \mathcal{B}} \beta$. An open problem is to determine whether {\em each} of the possible values returned by the algorithm is an upper bound on the diameter, i.e. whether $\beta$ is an upper bound on the diameter of the Cayley graph for each $\beta \in \mathcal{B}$. The examples studied so far show that for many families of trees (in fact, for all the trees investigated so far), the minimum possible value returned by the algorithm is also an upper bound on the diameter.  Thus, we believe that each of the possible values $\beta \in \mathcal{B}$ is an upper bound on the diameter of the Cayley graph, but we do not have a proof of this.

We now prove that the value returned by the algorithm is not necessarily unique:

\begin{Theorem} \label{thm:nonuniqueness:algA}
There exist transposition trees for which the value returned by Algorithm A is not unique.
\end{Theorem}

\begin{proof} Consider the transposition tree defined by the edge set $\{(1,2),(2,3),$ $(3,6),(3,4),(4,5),$ $(6,7),$ $(6,8),(6,9)\}$, which is shown in Figure~\ref{fig:tree:T2}. If Algorithm A picks the sequence of vertex pairs during step A2 to be $\{1,5\},\{2,7\},\{4,8\}$ and $\{3,9\}$, then the value returned by the algorithm is $\beta=7+5+5+3=20$.  And if Algorithm A picks the vertex pairs to be $\{1,7\},\{5,8\},\{2,9\}$ and $\{4,6\}$, then the value returned by the algorithm is $\beta=7+7+5+3=22$.  Hence, $\mathcal{B}$ contains $\{20,22\}$.
\end{proof}
\begin{center}
\begin{figure}
\begin{pspicture}(-3,-0.5)(5,2.5)
\qdisk(0,2){\vs}
\qdisk(0,0){\vs}
\qdisk(1,1.5){\vs}
\qdisk(1,0.5){\vs}
\qdisk(2,1){\vs}
\qdisk(3,1){\vs}
\qdisk(3,0){\vs}
\qdisk(3,2){\vs}
\qdisk(4,1){\vs}
\psline[linewidth=\et]{-}(0,2)(1,1.5)
\psline[linewidth=\et]{-}(0,0)(1,0.5)
\psline[linewidth=\et]{-}(1,1.5)(2,1)
\psline[linewidth=\et]{-}(1,0.5)(2,1)
\psline[linewidth=\et]{-}(2,1)(3,1)
\psline[linewidth=\et]{-}(3,1)(3,2)
\psline[linewidth=\et]{-}(3,1)(3,0)
\psline[linewidth=\et]{-}(3,1)(4,1)
\uput[270](0,0){$5$}
\uput[90](0,2){$1$}
\uput[90](1,1.5){$2$}
\uput[270](1,0.5){$4$}
\uput[90](2,1){$3$}
\uput[315](3,1){$6$}
\uput[90](3,2){$7$}
\uput[0](4,1){$8$}
\uput[270](3,0){$9$}
\end{pspicture}
\caption{The transposition tree used in the proof of Theorem~\ref{thm:nonuniqueness:algA} and Theorem~\ref{thm:beta:lessthan:fT:is:possible}.} \label{fig:tree:T2}
\end{figure}
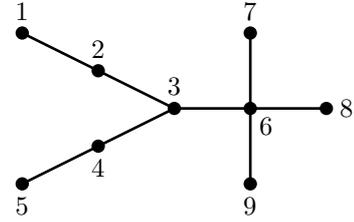
\end{center}

In addition to time complexity, another advantage of Algorithm A over the previous upper bound is that the value computed by Algorithm A can sometimes be \emph{strictly} less than the previous upper bound on the diameter.
\begin{Theorem} \label{thm:beta:lessthan:fT:is:possible}
The value computed by Algorithm A is always less than or equal to the previously known upper bound on the diameter, and there exist transposition trees for which the value computed by Algorithm A is strictly less than the previous upper bound.
\end{Theorem}

\begin{proof}
The first part of the assertion has been proved earlier.  For the second part, consider the transposition tree shown in Figure~\ref{fig:tree:T2}. It can be confirmed with the help of a computer that the true diameter value of the Cayley graph generated by this tree $T$ is 18 and the previous upper bound $f(T)$ on the diameter evaluates to 22.  As mentioned above, if Algorithm A picks the sequence of vertex pairs during step A2 to be $\{1,5\},\{2,7\},\{4,8\}$ and $\{3,9\}$, then the value returned by the algorithm is $\beta=7+5+5+3=20$, which is strictly less than the previous upper bound.
\end{proof}

In Theorem~\ref{thm:strictness} (cf. also the recent work Ganesan \cite{Ganesan:ICMMSC:2012}), it was shown that the difference between the previous upper bound $f(T)$ and the actual diameter $\diam(\Gamma)$ is at least $n-4$.  The proof given there is quite involved and required an examination of several (over 25) subcases.

As another advantage of Algorithm A over the previous upper bound $f(T)$, we show that the value computed by Algorithm A can also have a difference of at least $n-4$ from the actual diameter value but that the proof of this result is much simpler than the corresponding result for the previous upper bound $f(T)$:

\begin{Proposition} \label{prop:strictness:AlgA} For every $n$, there exists a transposition tree on $n$ vertices, such that the difference between the value computed by Algorithm A and the actual diameter value of the Cayley graph is at least $n-4$.
\end{Proposition}

\begin{proof}
Consider the transposition tree $\{(1,2),(2,3),\ldots,(n-3,n-2),(n-2,n-1),(n-2,n)\}$ shown in Figure~\ref{fig:tree:strictness:algA}.  The diameter of this tree is $n-2$.  After Algorithm A picks and removes two vertices from this tree that are a distance $n-2$ apart, we obtain the path graph on $n-2$ vertices.  The unique value computed by Algorithm A is thus $\{2(n-2)-1\} + \{2(n-3)-1\} + \{2(n-5)-1\}+\cdots$, which equals $\{2(n-2)-1\} + {{n-2} \choose 2} = {{n-1} \choose 2} + n-3$.

Now, any permutation can be sorted on this tree using at most ${{n-1} \choose 2}+1$ edges.  Indeed, marker 1 can be homed to its vertex using at most $n-2$ edges, and this vertex can then be removed from the tree.  Marker 2 can then be homed using at most $n-3$ edges, and so on, and marker $n-4$ can be homed using at most 3 edges.  At this point, we arrive at a star $K_{1,3}$, and any permutation on this star can be sorted using at most 4 edges since the diameter of the Cayley graph generated by this star is equal to 4.  Thus, the diameter of the Cayley graph generated by this tree is at most $(n-2)+(n-3)+\ldots+3+4 = {{n-1} \choose 2} +1$.

Hence, for this transposition tree, the difference between the value computed by Algorithm A and the actual diameter value is at least $n-4$.
\end{proof}

\begin{center}
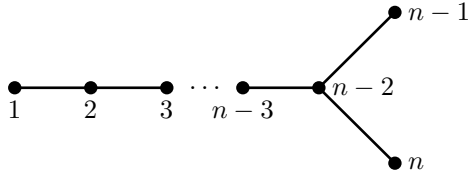
\begin{figure}
\begin{pspicture}(-3,0)(5,2)
\qdisk(0,1){\vs}
\qdisk(1,1){\vs}
\qdisk(2,1){\vs}
\qdisk(3,1){\vs}
\qdisk(4,1){\vs}
\qdisk(5,2){\vs}
\qdisk(5,0){\vs}
\psline[linewidth=\et]{-}(0,1)(1,1)
\psline[linewidth=\et]{-}(1,1)(2,1)
\uput[270](2.5,1.2){$\ldots$}
\psline[linewidth=\et]{-}(3,1)(4,1)
\psline[linewidth=\et]{-}(4,1)(5,2)
\psline[linewidth=\et]{-}(4,1)(5,0)

\uput[270](0,1){$1$}
\uput[270](1,1){$2$}
\uput[270](2,1){$3$}
\uput[270](3,1){$n-3$}
\uput[0](4,1){$n-2$}
\uput[0](5,2){$n-1$}
\uput[0](5,0){$n$}
\end{pspicture}
\caption{The transposition tree used in the proof of Proposition~\ref{prop:strictness:AlgA}} \label{fig:tree:strictness:algA}
\end{figure}
\end{center}

Note that Theorem~\ref{thm:betamax:ubound} implies that the value computed by Algorithm A is an upper bound on the diameter for all trees for which $|\mathcal{B}|=1$ since for such trees any value $\beta$ computed by the algorithm is equal to $\beta_{\max}$.  For such trees, Algorithm A efficiently computes a value which is both an upper bound on the diameter as well as better than (or at least as good as) the previously known diameter upper bound.   In the rare case when the value computed by Algorithm A is not unique, it still seems likely that that each of the possible values computed by the algorithm is an upper bound on the diameter, though we do not have a proof for this.

\section{Concluding remarks}

Cayley graphs have been studied as a suitable model for the topology of interconnection networks, and a problem of both theoretical and practical interest is to obtain bounds for the diameter of Cayley graphs.  In this work, we investigated an upper bound on the diameter of Cayley graphs generated by transposition trees.   We showed that for every $n >4$, there exists a tree on $n$ vertices such that the difference between the previous upper bound $f(T)$ and the true diameter is at least $n-4$.  Such results are of interest because they give us insight as to how far away these bounds can be from the true diameter value in the worst case and sometimes tell us for which families of graphs this bound can be utilized or not utilized.

The $n-4$ lower bound on $\Delta_n$ was seen to be best possible in the sense that it is attained for some values of $n$. Now consider the tree on 9 vertices consisting of the edge set $\{(1, 2), (2, 3),$ $(3, 4), (4, 5),$ $(5, 6), (6, 7),$ $(6, 8), (6, 9)\}$. Then, it can be confirmed with the help of a computer that the diameter of the Cayley graph generated by this tree is 24, and the diameter upper bound $f(T )$ for this tree evaluates to 30. Hence, this $n-4$ lower bound is not the exact value of the strictness $\Delta_n$, and an open problem is to obtain an upper bound for $\Delta_n$.

An efficient algorithm (Algorithm A) for the diameter of Cayley graphs generated by transposition trees was given. Remarkably, the algorithm has time complexity $O(n^2)$, compared to the previous upper bound in the literature $f(T)$ for which no polynomial time algorithm is known and for which the existing methods take time $\Omega(n!)$ to compute.   We proved  that the value obtained by our algorithm is less than or equal to the previously known diameter upper bound. Sometimes the value computed by the algorithm is {\em strictly} less than the previous upper bound.  Such an improvement in efficiency from time $\Omega(n!)$ to polynomial time $O(n^2)$, while still performing at least as well as the previous bound, was possible because we worked directly with the transposition tree on $n$ vertices, and so our algorithm does not require examining any of the permutations.

We described some further advantages of our algorithm over the previous bound, besides the improvement in time complexity.  We provided a tree for which the value computed by the  algorithm is not necessarily unique.  This is an important fact because such counterexamples are quite rare.

We believe that each of the possible values computed by Algorithm A is an upper bound on the diameter of the Cayley graph, but we do not have a proof for this.  For the families of trees investigated so far, the maximum possible value returned by our algorithm is exactly equal to the previously known diameter upper bound. However, our algorithm arrived at the same value using a very different (and also simpler and more efficient) method than the previously known diameter upper bound, and investigating further properties of this algorithm might lead to new insights on this problem.

The algorithm presented here raises many further interesting questions and problems.  For example, is it true that the algorithm returns a unique value (i.e. $|\mathcal{B}|=1$) for almost all trees? Is it true that for all trees, the maximum possible value $\beta_{\max}$ returned by the algorithm is equal to $f(T)$?  Other open problems include to characterize those trees for which the value returned by Algorithm A is unique, and to characterize those trees for which the sequences of subtrees generated by the algorithm are isomorphic (i.e. are independent of the choice of vertex pairs during step A2, unlike the tree in Example 1 above).

\ifCLASSOPTIONcaptionsoff
  \newpage
\fi



%




\bibliographystyle{BibTeXtran}   
\bibliography{refsaut}       

%


\begin{IAENGbiographynophoto}{Ashwin Ganesan}
was born in Kalpakkam, Tamil Nadu, India, in 1977. He received the B.S. degree in electrical engineering from Marquette University, Milwaukee, WI, USA, in 1998, and the M.S. degree in electrical engineering from the University of Wisconsin at Madison, WI, USA, in 2000.  His research interests are in combinatorics, graph theory, algorithms, discrete mathematics, and their applications.

He received the Top Scholar in Curriculum Award from Marquette University in 1998, the Frank Rogers Bacon Fellowship from the University of Wisconsin at Madison, WI, USA during 1998-1999, and the Regents' Fellowship from the University of California at Berkeley, CA, USA, during 2001-2002.  He was a Teaching and Research Assistant at the University of California at Berkeley during 2001-2004 and at the University of Wisconsin at Madison during 2004-2008.  He was a Senior Lecturer and Assistant Professor at Mumbai University affiliated engineering colleges during 2008-2010.  Since 2010, he has been an Assistant Professor in the Department of Mathematics, Amrita School of Engineering, Amrita University, Tamil Nadu, India.

\end{IAENGbiographynophoto}






\end{document}